\documentclass[aps,pra,superscriptaddress,notitlepage,twocolumn]{revtex4-1}
\usepackage[utf8]{inputenc}
\usepackage{braket}
\usepackage{amsmath,amsfonts}
\usepackage{graphicx}
\usepackage{bm}
\usepackage{amsthm}
\usepackage[colorlinks]{hyperref}
\usepackage{xcolor}
\usepackage{comment}
\usepackage{mathtools}
\usepackage{physics}
\usepackage{multirow}

\newtheorem{theorem}{Theorem}

\newtheorem{lemma}{Lemma}
\newtheorem{defi}{Definition}
\newcommand{\ba}{\begin{eqnarray}}
\newcommand{\ea}{\end{eqnarray}}
\newcommand{\ban}{\begin{eqnarray*}}
\newcommand{\ean}{\end{eqnarray*}}
\newcommand{\be}{\begin{equation}}
\newcommand{\ee}{\end{equation}}

\def\cH{{\cal H}}
\def\sc{\mathop{\rm sc}}
\def\gcm{\mathop{\rm gcm}}
\def\Label#1{\label{#1}\ [\ \text{#1}\ ]\ }
\def\Label{\label}

\begin{document}

\title{Measurement-Device-Independent Detection of Beyond-Quantum State}

\author{Baichu Yu}
\email{yubc@sustech.edu.cn}
\affiliation{Shenzhen Institute for Quantum Science and Engineering, 
Southern University of Science and Technology, Nanshan District, Shenzhen, 518055, China}
\affiliation{International Quantum Academy (SIQA), Shenzhen 518048, China}

\author{Masahito Hayashi}
\email{hmasahito@cuhk.edu.cn}
\affiliation{School of Data Science, The Chinese University of Hong Kong, Shenzhen, Longgang District, Shenzhen, 518172, China}
\affiliation{International Quantum Academy (SIQA), Shenzhen 518048, China}
\affiliation{Graduate School of Mathematics, Nagoya University, Nagoya, 464-8602, Japan}

\date{}

\begin{abstract}
In quantum theory, a quantum state on a composite system 
of two parties
realizes a non-negative probability with 
any measurement element with a tensor product form.
However, there also exist non-quantum states which satisfy the above condition.
Such states are called beyond-quantum states, and cannot be detected by standard Bell tests. 
To distinguish a beyond-quantum state from quantum states,
we propose a measurement-device-independent (MDI) test for beyond-quantum state detection, which is composed of 
quantum input states on respective parties and
quantum measurements across the input system and the target system on respective parties. 
The performance of our protocol is independent of 
the forms of the tested states and the measurement operators, which provides an advantage in practical scenarios.
We also discuss the importance of
tomographic completeness of the input sets to the detection.
\end{abstract}

\maketitle

\section{introduction}
Bell test is a very important type of protocol to study the theoretical and practical problems of quantum information. Standard Bell test contains two spatially separated parties, each of which are given some classical inputs and required to generate corresponding classical outputs. The result of the standard Bell test only depends on the generated input-output statistics, which are referred to as correlations, and no knowledge about the experimental setups or the physical theory is required. Such a property is called device-independent (DI). The DI property of Bell test makes it a convenient scenario for studying theories which could be more general than quantum mechanics. Some of these theories are proposed based on physical principles \cite{linden2007quantum,pawlowski2009information,navascues2010glance}, some of them are based on more general mathematical frameworks \cite{navascues2015almost,barnum2010local,arai2023detection}. 

General Probabilistic Theory (GPT) \cite{plavala2023general} is a theory equipped with general states and measurements, which together produce probability distribution for experimental outcomes. GPT requires only two weak assumptions: the independence between the choice of state and the choice of measurement, and the independent and identically distributed (i.i.d) property of data generated by repeating experiment \cite{grabowecky2022experimentally}, therefore it includes quantum theory. Due to its generality, GPT is used to discuss the mathematical structure \cite{arai2023pseudo} and information processing properties \cite{barrett2007information,dall2017no} of quantum theory. In particular, there are models which possess the same local structure as quantum theory, but has different global structure \cite{stevens2014steering,arai2019perfect,aubrun2022entanglement,arai2023pseudo}. 

Then a natural question would be whether we can identify quantum mechanics with the correlations generated in Bell-type DI scenarios. The negative answer was shown by the existence of a family of GPT states which is more general than quantum states, but produces only quantum correlations in DI Bell tests \cite{barnum2010local,arai2023detection}. Such states are locally quantum and generate valid probability distribution under local (separable) measurements, therefore are called positive over all pure tensors (POPT) states \cite{klay1987tensor,barnum2005influence}. The set of POPT states includes the set of quantum states, and also some non-quantum states which have negative eigenvalues. 
In the following, we refer to such non-quantum POPT states as beyond-quantum states. 
Briefly speaking, DI Bell tests can not detect beyond-quantum states 
on the composite system of $A$ and $B$
due to the following reasons.
Any beyond-quantum state $\rho_{AB}^{nq}$  can be 
obtained by performing a positive trace-preserving map 
$\Gamma_A$
on $A$ to a quantum state $\rho_{AB}^{q}$.
The dual map $\Gamma_A^*$ of the positive map $\Gamma_A$ is also a positive map, which turns a 
Positive-Operator-Valued Measurement (POVM) element
$M^A_a$ on $A$ 
into another POVM element $\Gamma_A^*(M^A_a)$ on $A$. 
Therefore, the correlations produced by 
the pair of a beyond-quantum state $\rho_{AB}^{nq}$ 
and a local measurement $\{M^A_a\otimes M^B_b\}_{a,b}$ 
can be simulated by the correlations produced by 
the pair of the above quantum state $\rho_{AB}^{q}$
and a local measurement $\{\Gamma_A^*(M^A_a)\otimes M^B_b\}_{a,b}$ .
That is, if the measurement operators can not be certified,
the correlation produced by measuring a beyond quantum
state cannot be distinguished from the correlation
produced by measuring a quantum state. 

In order to detect/identify beyond-quantum states with local measurements, we need some protocol which is more restrictive on the measurement operators than DI Bell test. A straightforward idea is to consider device-dependent (DD) protocol, where the form of the measurement operators are known \cite{arai2023detection}. 
However, we do not always have so much knowledge on measurement devices, e.g.,
sometimes the measurement devices can not be trusted. In those cases, the soundness of the result can not be guaranteed. Therefore we consider a type of an intermediate protocol between DI and DD ones, 
called a measurement-device-independent (MDI) protocol. 
The key property of an MDI protocol is that it transfers the trust on measurement devices into the trust on state preparations, such that the soundness of the protocol is independent of the knowledge of measurement operators. Consequently, MDI protocols are widely used in tasks such as quantum key distribution (QKD) \cite{lo2012measurement,liu2013experimental,curty2014finite} and entanglement detection \cite{branciard2013measurement,rosset2018practical,guo2020measurement}. 

The aim of this paper is 
to propose an MDI protocol for beyond-quantum state detection, which is composed of 
quantum input states on respective parties and
quantum measurements across the input system and the target system on respective parties as Fig. \ref{fig1}. 
To address the requirement for our protocol, 
we will first discuss the following three criteria about the performance of the beyond-quantum state detection protocol, which we name as completeness, universal completeness and soundness, 
while the words completeness and soundness were originally introduced in the context of 
the verification \cite{hayashi2018self,li2023robust}.
Based on these concepts, we will propose our MDI detection protocol 
for beyond-quantum state.

\begin{figure}[t]
    \centering
  \includegraphics[width=1\linewidth]{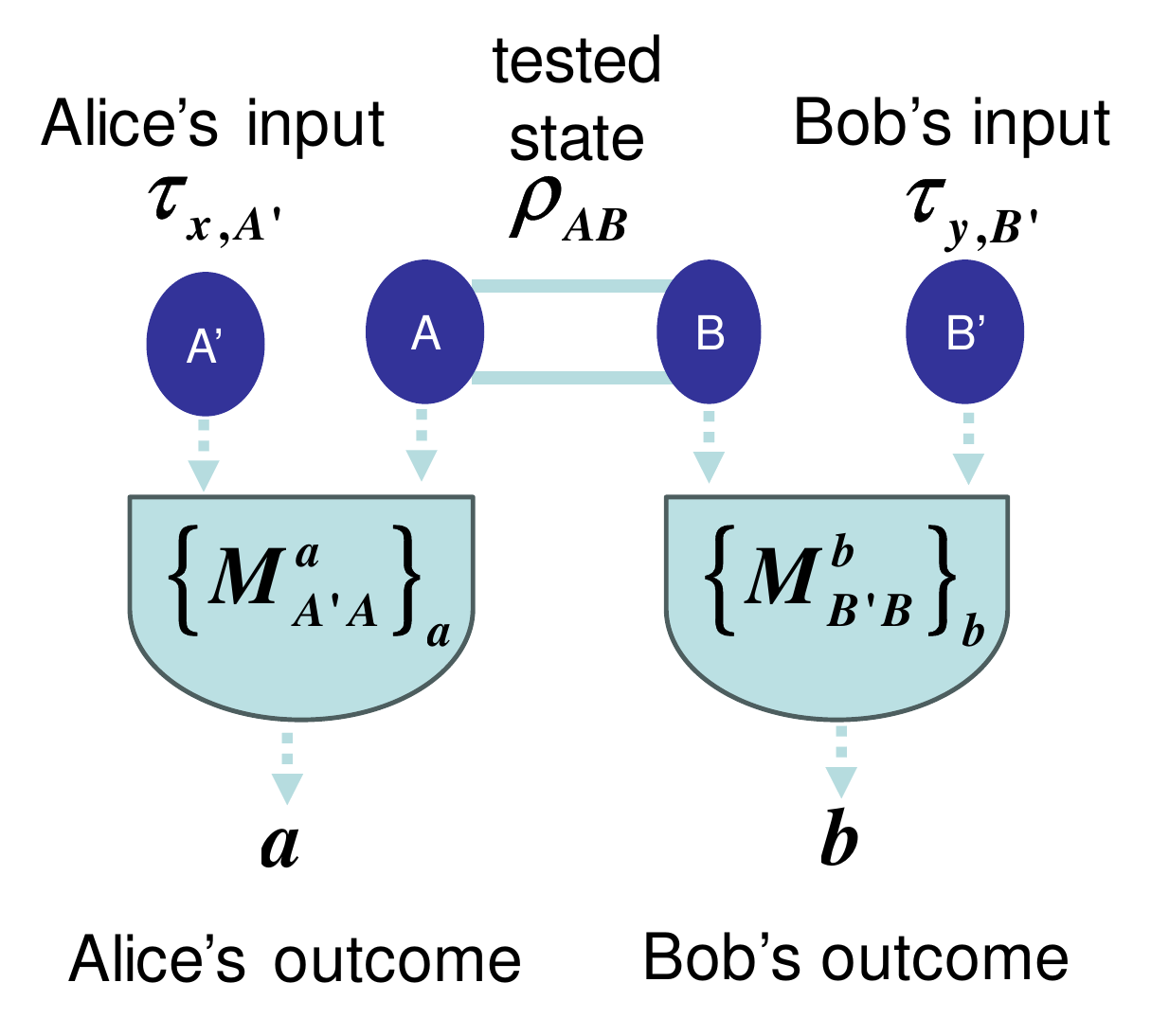}
  \caption{MDI protocol.
An MDI protocol is composed of 
the target state $\rho_{AB}$,
quantum input states $\{\tau_{x,A'}\}_x$, $\{\tau_{y,B'}\}_y$, and
quantum measurements 
$M_{A'A}=\{M_{A'A}^{a}\}_a$ and $ M_{BB'}=\{M_{BB'}^{b}\}_b$
across the input system and the target system
in respective parties.  
As the result of this experiment,
the MDI protocol generates the correlation set 
$\{p(a,b|\tau_{x,A'},\tau_{y,B'})\}$. 
The details of these notations will be given in Section \ref{S4}. 
}
  \label{fig1}
  \end{figure}

\begin{defi}
\emph{Completeness}: 
Let $S$ be a set beyond-quantum states.
A protocol is called complete for the set $S$
when 
any beyond-quantum state in $S$ can be detected by the protocol 
under the assumption that 
the experimental devices are properly chosen.
\end{defi}

\begin{defi}
\emph{Universal completeness}: 
A protocol is called universally complete 
when it is complete for the set of
all beyond-quantum states (with certain fixed local dimensions).
\end{defi}

\begin{defi}
\emph{Soundness}: Any quantum state will never be detected as a beyond-quantum state by the protocol.
\end{defi}

The remaining of this paper is organized as follows.
Section \ref{S2} discusses the relation between our protocol and existing protocols.
Section \ref{S3} gives definitions of
a beyond-quantum state, a MDI protocol, and its special cases.
Section \ref{S4} presents our protocol.
Section \ref{S5} derives a necessarily and sufficient condition for
our protocol being universally complete.

\section{Relation with existing protocols}\label{S2}
First, the studies
\cite{branciard2013measurement,buscemi2012all,abiuso2021measurement}
introduced a MDI Bell test 
as a type of a generalized Bell test
by replacing the classical inputs of the standard Bell test by quantum input states. 
MDI Bell test is not fully DI since it assumes quantum theory locally and requires the knowledge (trust) on the quantum input states. The additional assumptions provides a MDI Bell test some advantages over DI Bell test, e.g. it does not suffer the locality or detection loophole \cite{branciard2013measurement,rosset2018practical}. Moreover, it certifies certain properties of measurement operators (such as entanglement). 

In a recent work, it was shown that a MDI Bell test can be used to detect beyond-quantum states \cite{lobo2022certifying}. More specifically, the authors proved that for any given beyond-quantum state, it is possible to construct a MDI witness, which is a linear function of the correlations 
$\{p(a,b|\tau_{x,A'},\tau_{y,B'})\}_{a,b,x,y}$. The construction of the witness is dependent on the form of the tested state and the measurement operators. However, in a practical MDI Bell test, the form of the tested state or measurement operators is usually not assumed to be known. In that case, the proper witness for detection can not be determined, which may greatly influence the detectability (the ability of detecting a given beyond-quantum state) of the protocol. 

The difference between our protocol and that in \cite{lobo2022certifying} lies in the method of processing the experimental correlations. 
In our protocol, the detectability depends only on the experimental correlations, but not the forms of the tested state and the measurement operators. Such property provides our protocol an advantage when  
the tested state and the measurement operators are unknown.
Here, we examine whether the existing detection protocols for beyond-quantum state satisfy the above criteria, and what assumptions they require. 
For the device-dependent protocol proposed in Ref.~\cite{arai2023detection}, 
the knowledge of the forms of the measurement operators
is required to satisfy soundness, and
the form of the tested state is additionally required to satisfy completeness. 
For the MDI protocol proposed in Ref.~\cite{lobo2022certifying}, 
the knowledge of the form of quantum input states is required to satisfy soundness,
and 
the knowledge of the forms of the tested state and the measurement operators is additionally required to satisfy completeness.
For our protocol, while the knowledge of the form of quantum inputs
is also required to satisfy soundness, 
only the knowledge of the local dimensions of the tested state is additionally required to satisfy completeness. Also, since our protocol satisfies completeness without the knowledge of the tested state and the measurement operators, it satisfies universal completeness. In Table~\ref{table}, we summarize the above differences in a table.

\begin{table}[t]
\resizebox{\linewidth}{!}{
\begin{tabular}{|c|c|c|}
\hline &Soundness& Completeness
\\ \hline
\multirow{3}{*}{Protocol in Ref.\cite{arai2023detection}} 
& forms of
& \multirow{2}{*}{+ form of} 
\\
& measurement 
& \multirow{2}{*}{the tested state}
\\
& operators
& 
\\ \hline
\multirow{3}{*}{Protocol in Ref.\cite{lobo2022certifying}} 
& forms of
& + forms of
\\ 
& quantum
& the tested state and 
\\ 
& input states 
& measurement operators
\\ 
\hline
\multirow{3}{*}{Our protocol}
& forms of 
&\multirow{2}{*}{+ dimension of} 
\\ 
&quantum  
&\multirow{2}{*}{the local systems}
\\ 
&input states 
&
\\ \hline
\end{tabular}
}
\caption{Requirement to satisfy the two criteria,
soundness and completeness. 
In the table, we summarize the knowledge needed to satisfy soundness and completeness for the three protocols. The ``+" sign means the condition is additionally required besides the condition for soundness. It can be seen that our protocol satisfies the two criteria with relatively weaker assumptions.}
\label{table}
\end{table}

\section{Beyond-quantum state and MDI protocol}\label{S3}
At first we briefly introduce the setting of a GPT model in the form used in Ref.~\cite{arai2023detection}. 

\begin{defi}
A model of GPT is defined by a tuple $G=(\mathcal{V},\langle,\rangle,\mathcal{C},u)$, where $(\mathcal{V},\langle,\rangle),\mathcal{C},u$ are a real-vector space with inner product, a proper cone, and an order unit of $C^{*}$ respectively.
\end{defi}

\begin{defi}
 The space $S(G)$ given the GPT model $G=(\mathcal{V},\langle,\rangle,\mathcal{C},u)$ is defined as
\begin{equation}
S(G):=\{\rho\in \mathcal{V}|\langle\rho,u\rangle=1 \}.
\end{equation}
\end{defi}

\begin{defi}
A measurements is a family $\{M_{i}\}$ satisfying $M_{i}\in C^{*}$ and $\sum_{i}M_{i}=I$. Here $i$ is called the outcome of a measurement, and the probability of obtaining outcome $i$ is $p_{i}=\langle\rho,M_{i}\rangle$. 
\end{defi}

In this work we consider a less general but very practical GPT model, where each measurement operator can be written as a Hermitian matrix, or equivalently, $dim \mathcal{V}=d^2$ for some integer $d>1$, then each state can be written as a Hermitian matrix of unit trace, and the inner product takes the form of trace \cite{gleason1975measures,busch2003quantum,arai2023derivation}. Also, in this work, 
our scenario has two spatially separated experimenters, Alice and Bob, who can only perform local measurements. Let $\mathcal{L}_{h}(\mathcal{H}_{S})$ denote the set of Hermitian matrices on a finite dimensional Hilbert space $\mathcal{H}_{S}$, sometimes we also simply write it as $\mathcal{L}_{S}$ and write $\mathcal{L}_{h}(\mathcal{H}_{S}\otimes\mathcal{H}_{V})$ as $\mathcal{L}_{SV}$. Let $\mathcal{L}^{+}$ denote the positive semi-definite matrices among the corresponding set $\mathcal{L}$. Then the GPT model we consider in this work can be written as $G=(\mathcal{L}_{AB},\textrm{Tr}, SEP(A:B)^{*},I)$, where $SEP(A:B)^{*}:=(\mathcal{L}_{A}^{+}\otimes\mathcal{L}_{B}^{+})^{*}$ is the cone of (unnormalized) POPT states.

Now we formally introduce a MDI protocol and a MDI witness as its special case.
Their aim is to distinguish
two sets:
the set $\mathcal{Q}_{AB}$ of quantum states in $\mathcal{L}_{AB}$ 
and
the set $\mathcal{S}_{AB}$ of states in $SEP(A:B)^*$, i.e.,
to discriminate 
an element of $\mathcal{S}_{AB}\setminus \mathcal{Q}_{AB}$
from the set $\mathcal{Q}_{AB}$.

Similar to entanglement detection \cite{branciard2013measurement}, 
a MDI protocol 
for a joint state $\rho_{AB}\in \mathcal{Q}_{AB}$ across Alice's and Bob's systems is given as follows.
In each round, Alice and Bob are given one of the quantum states in the sets 
$\{\tau_{x,A'}\}_x$, $\{\tau_{y,B'}\}_y$ on 
auxiliary systems $A'$ and $B'$, respectively.
They make quantum measurements 
$M_{A'A}=\{M_{A'A}^{a}\}_a$ and $ M_{BB'}=\{M_{BB'}^{b}\}_b$ on system $A'A$ and $BB'$ respectively, obtaining outcomes $a$ and $b$. The generated correlations is obtained by
\begin{equation}\Label{correlations}
p(a,b|\tau_{x,A'},\tau_{y,B'})=\textrm{Tr} (\tau_{x,A'}\otimes \rho_{AB}\otimes \tau_{y,B'}\cdot M_{A'A}^{a}\otimes M_{BB'}^{b}).
\end{equation}
A MDI protocol outputs a certain outcome depending only on 
the correlations $P:=\{p(a,b|\tau_{x,A'},\tau_{y,B'})\}_{a,b,x,y}$.
During the process, Alice and Bob is allowed to communicate classically, and they can also post-process their measurement correlations. These actions do not affect the soundness of the detection.

A MDI witness is a special type of a MDI protocol,
and is constructed by using the correlations as below:
\begin{equation}\label{MDIwit}
\mathcal{W}_{\beta}(P)
=\sum_{x,y,a,b} \beta_{x,y,a,b}p(a,b|\tau_{x,A'},\tau_{y,B'}),
\end{equation}
where $\beta=(\beta_{x,y,a,b})_{x,y,a,b}$ 
are some real coefficients dependent on $x,y,a,b$. When the input set $\mathcal{I}:=\{\tau_{x,A'}\otimes\tau_{y,B'}\}_{x,y}$ forms a tomographically complete set on $\mathcal{L}_{AB}$, 
a MDI witness can correspond to an entanglement witness as follow: for any entanglement witness $W\in\mathcal{L}_{AB}$, it can be spanned using $\tau_{x,A}:=\tau_{x,A'}^{\mathrm{T}}$, $\tau_{y,B}:=\tau_{y,B'}^{\mathrm{T}}$ as 
\begin{equation}\label{cdoe}
    W=\sum_{x,y}\alpha_{x,y} \tau_{x,A}\otimes \tau_{y,B}
\end{equation}
where $\alpha_{xy}$ are some real coefficients. 
To recover the above entanglement witness $W$ from 
a MDI witness, we fix 
one element $a_0$ among Alice's outcomes
and one element $b_0$ among Bob's outcomes,
and choose the coefficients $\beta_{x,y,a,b}$ to be
\begin{equation}\label{ccoe}
  \beta_{x,y,a,b}:=\left\{
  \begin{array}{ll}
      \alpha_{x,y}, &\hbox{when } a=a_0,b=b_0\\
      0, &\hbox{otherwise.}
                       \end{array} 
                       \right.
\end{equation}
Then, we have $\mathcal{W}_{\beta}(P)\geq 0$ 
when the tested state $\rho_{AB}$ is separable. Also, when measurement operators take the form $M_{A'A}^{a_0}=\ketbra{\Phi_{A'A}}{\Phi_{A'A}}$, where 
$\ket{\Phi_{A'A}}:=\frac{1}{\sqrt{d_{A}}}\sum_{i=1}^{d_{A}}\ket{i}\ket{i}$ 
is the maximally entangled state in $\mathcal{H}_{A'}\otimes\mathcal{H}_{A}$, and analogously $M_{BB'}^{b_0}=\ketbra{\Phi_{BB'}}{\Phi_{BB'}}$, we have 
$\mathcal{W}_{\beta}(P)=\frac{1}{d_{A}d_{B}}\Tr (W\rho_{AB})$ \cite{branciard2013measurement}. Therefore entanglement is detected if 
$\mathcal{W}_{\beta}(P)<0$.

In Ref.~\cite{lobo2022certifying}, the authors showed that MDI witness can also be used to detect beyond-quanutm states. The central point is that there exists a positive semi-definite matrix $Y$ for every beyond-quantum states $\rho_{AB}\in\mathcal{S}_{AB}$, such that 
\begin{align}\label{witnessY}
\begin{split}
\textrm{Tr}\rho_{AB} Y& <0, \\
\textrm{Tr}\sigma_{AB} Y& \ge 0
\hbox{ for } \forall \sigma_{AB}\in \mathcal{Q}_{AB}. 
\end{split}
\end{align}
We will refer to such matrices $Y$ as the beyond-quantumness witness (or simply witness if it does not cause confusion) of the beyond-quantum state $\rho_{AB}$ hereinafter. By substituting the entanglement witness $W$ in the above discussion with the beyond-quantumness witness $Y$ and constructing the coefficients $\beta_{x,y,a,b}$ analogously, the MDI entanglement witness in \eqref{MDIwit}  
becomes a MDI beyond-quantum state witness. Ref.~\cite{lobo2022certifying} reveals the strong potential MDI protocol has on detecting beyond-quantum states, but there are still problems to address when we come to a more practical scenario. More specifically, to witness a beyond-quantum state with \eqref{MDIwit}, the choice of suitable coefficients $\beta_{x,y,a,b}$ is dependent on the knowledge of the form of the tested state $\rho_{AB}$ and the measurement operators $M_{A'A}^{a}$ and $ M_{BB'}^{b}$. 

To see this dependence, for example, we focus on a two-qubit system
with two Bell states
$\ket{\Phi_{+}}:=\frac{1}{\sqrt{2}}(\ket{00}+\ket{11})$,
$\ket{\Psi_{-}}:=\frac{1}{\sqrt{2}}(\ket{01}-\ket{10})$
and the partial transposition 
$\Gamma_{A}:=\textrm{T}_{A}\otimes I_{B}$
on subsystem $A$, where
$\textrm{T}_{A}$ expresses the transposition
on subsystem $A$.
The witness $Y_{0}=\ketbra{\Psi_{-}}{\Psi_{-}}$ 
of $\rho_{0}:=\Gamma_{A}(\ketbra{\Phi_{+}}{\Phi_{+}})$
can be recovered by a MDI witness as follows.
When we choose $\beta_{x,y,a,b}$ as \eqref{ccoe} with $\alpha_{x,y}$ satisfying
\begin{equation}
    Y_{0}=\sum_{x,y}\alpha_{x,y} \tau_{x,A}\otimes \tau_{y,B}
\end{equation}
and define the correlations $P=\{p(a,b|\tau_{x,A'},\tau_{y,B'})\}_{a,b,x,y}$ by \eqref{correlations}
with the measurement operators $M_{A'A}^{a_0}=M_{BB'}^{b_0}=\ketbra{\Phi_{+}}{\Phi_{+}}$,
our MDI witness $\mathcal{W}_{\beta}$
satisfies 
$\mathcal{W}_{\beta}(P)=\frac{1}{4}\Tr (Y_{0}\rho_{0})
$,
and
therefore detects $\rho_{0}$.
However, under the same setting (input set and measurement operators), 
the same MDI witness $\mathcal{W}_{\beta}$
can not detect the beyond-quantum state $\rho_{1}:=\Gamma_{A}(\ketbra{\Psi_{-}}{\Psi_{-}})$.
In practical MDI protocols, it is more common that we only know the local dimensions of $\rho_{AB}$ instead of their exact forms. Therefore, more advanced method is needed to address this problem.

Inspired by a previous work \cite{rosset2018practical}, we propose a method which processes the correlations obtained in a MDI detection experiment using Semidefinite Programming (SDP). We will show that this method satisfies universal completeness with only the knowledge about the dimensions of the local systems, and its detection is not influenced by the knowledge of the states and measurement operators. 

\section{MDI protocol with SDP}\label{S4}
An important observation which inspires our approach is that, although the knowledge of the tested state and the measurement operators are assumed to be known when constructing a MDI witness, neither of them is actually necessary. By processing Eq.~\eqref{correlations} we obtain the relation
\begin{equation}\Label{piab0}
\begin{split}
p(a,b|\tau_{x,A'},\tau_{y,B'})
=\Tr_{A'B'}(\tau_{x,A'}\otimes\tau_{y,B'} \Pi_{ab}),\ \forall a,b,x,y,
\end{split}
\end{equation}
where
\begin{equation}\Label{piab}
\Pi_{ab}=\Tr_{AB}(I_{A'}\otimes \rho_{AB}\otimes I_{B'} M^{a}_{A'A}\otimes M^{b}_{BB'})
\end{equation}
is some operator corresponding to $a,b$ in $\mathcal{L}_{A'B'}$. Eq.~\eqref{piab0} allows us to reduce the detection of beyond-quantumness onto a set $\{\Pi_{ab}\}_{a,b}$ of operators, 
where each $\Pi_{ab}$ could be understood as an effective state (unnormalized), and $\tau_{x,A'}$ and $\tau_{y,B'}$ serve as effective measurement operators. It can be obtained from Eq.~\eqref{piab} that $\Pi_{ab}\geq 0$ when $\rho_{AB}\geq 0$. Therefore if any $\Pi_{ab}$ is found to be non-positive semi-definite, the tested state is beyond-quantum.  

When the form of the tested state or the measurement operators are not known, 
the form of $\Pi_{ab}$ are also not known. In such a case, we propose 
an SDP for beyond-quantum state detection as below:\\
\begin{equation}\Label{SDP}
\begin{aligned}
\text{minimise:}\ &c_{ab}\\
\text{such that:}\ &c_{ab}=\text{Tr}X_{ab}^-\\
                &X_{ab}=X_{ab}^+-X_{ab}^-\\
                 &  X_{ab}^+\geq 0, X_{ab}^-\geq 0,\\
                 &  \textrm{Tr}(X_{ab}\tau_{x,A'}\otimes\tau_{y,B'})=p(a,b|\tau_{x,A'},\tau_{y,B'}), \forall x,y.
\end{aligned}
\end{equation}
SDP \eqref{SDP} aims to reconstruct the least negative operators $X_{ab}\in\mathcal{L}_{A'B'}$ (with respect to the sum of negative eigenvalues) which can generate the experimental correlations via Eq.~\eqref{piab0}. It is not difficult to see that $c_{ab}>0$ if and only if the projection of $\Pi_{ab}$ onto the subspace spanned by $\{\tau_{x}\otimes\tau_{y}\}_{x,y}$ 
is non-positive semi-definite. Therefore if there exists a pair of outcomes $(a,b)$ such that $c_{ab}>0$, we conclude that the tested state is beyond-quantum. 
In SDP \eqref{SDP}, only the knowledge of the form of inputs $\tau_{x,A'}$ and $\tau_{y,B'}$ is used together with the correlations to obtain the result
because $c_{ab}$ is given as a function of inputs $\tau_{x,A'}$, $\tau_{y,B'}$, and
the observed correlation $p(a,b|\tau_{x,A'},\tau_{y,B'})$.
The detectability of our method is independent of the knowledge of the tested state and the measurement operators since it exhausts all possible operators which can generate the correlations. So our method is also optimal in the sense that if it does not detect the beyond-quantumness with given correlations, no other method can detect it.

Now we prove that when the sets of quantum inputs 
$\{\tau_{x,A'}\}_x$ and $\{\tau_{y,B'}\}_y$ are well chosen and known, \eqref{SDP} satisfies soundness. 
Without loss of generality,
we can always assume that the dimension of the auxiliary systems is the same as the corresponding local systems, i.e., $d_{A}=d_{A'}$ and $d_{B}=d_{B'}$.

\begin{theorem}\label{Thm1}
When 
$M_{A'A}$ and $ M_{BB'}$ 
are quantum measurements 
on system $A'A$ and $BB'$, 
and
$\{\tau_{x,A'}\}_x$ and $\{\tau_{y,B'}\}_y$ are chosen to be sets of states in $\mathcal{H}^{d_{A'}}$ and $\mathcal{H}^{d_{B'}}$, where $d_{A'},d_{B'}$ are the dimensions of the auxiliary systems $A'$ and $B'$, the tested state is beyond-quantum whenever we obtain $c_{ab}>0$ for some outcomes $a,b$ with SDP \eqref{SDP}. 
\end{theorem}

\emph{Proof.} 
If the tested state $\rho_{AB}$ is a quantum state, any effective states $\Pi_{ab}$ is positive semi-definite. Since the constraints of SDP \eqref{SDP} can always be satisfied by taking $X_{ab}=\Pi_{ab}$, the SDP will always output $c_{ab}=0$ in such case.   
This indicates that $c_{ab}>0$ only if $\rho_{AB}$ is beyond quantum. \qed

Therefore, whenever the form of the quantum input states $\{\tau_{x,A'}\}_x$ and $\{\tau_{y,B'}\}_y$ are known, 
our detection protocol satisfies soundness.
Note that in the proof of Theorem 1, we assume that Alice and Bob are honest: they do not communicate, nor do they process their local systems or the experimental correlations. 
Such an assumption does not affect the generality, since allowing local operation and classical communication (LOCC) between Alice and Bob will not produce false detection when the tested state is quantum. We leave the proof in the appendix \ref{appendixA}.

Now we show that SDP \eqref{SDP} satisfies universal completeness with the knowledge of the dimension of the local systems.

\begin{theorem}\label{Thm2}
Given the assumption that the dimensions of $A$ and $B$ are $d_A$ and $d_B$ respectively, any beyond-quantum state $\rho_{AB}$ can be detected by our MDI protocol 
when
$\{\tau_{x,A'}\}_x$ and $\{\tau_{y,B'}\}_y$ are tomographically complete sets of states on $\mathcal{H}^{d_A}$ and $\mathcal{H}^{d_B}$, and measurement operators $M^{a}_{A'A}$ and $M^{b}_{BB'}$ are entangled pure, with Schmidt rank $d_A$ and $d_B$ respectively. 
\end{theorem}

Theorem \ref{Thm2} shows that our protocol satisfies universal completeness with the knowledge of the form of quantum input states $\{\tau_{x,A'}\}_x$ and $\{\tau_{y,B'}\}_y$ and the dimension of local systems.

\emph{Proof.} We will first show that when the measurement operators $M^{a}_{A'A}$ and $M^{b}_{BB'}$ 
are chosen to have the form of generalized Bell measurement operators, the eigenvalues of $\Pi_{ab}$ will be proportional to the eigenvalues of $\rho_{AB}$. 
Generalized Bell measurement (GBM) operators has the form $M^{a}_{A'A}=\ketbra{\Lambda_a}{\Lambda_a}$, where $\ket{\Lambda_a}=U_{A'}^{a}\otimes I_{A} \ket{\Phi_{A'A}}$, $U_{A'}^{a}$ is some unitary matrix on $\cal{H}_{A'}$, and $\ket{\Phi_{A'A}}=\frac{1}{\sqrt{d_A}}\sum_{i}^{d_A}\ket{i}\ket{i}$ is the maximally entangled state. Then we have
\begin{equation}\Label{maxent}
M^{a}_{A'A}=\frac{1}{d_{A}}\sum_{i,j=1}^{d_A}U_{A'}^{a}\ketbra{i}{j}U_{A'}^{a\dagger}\otimes\ketbra{i}{j},
\end{equation}
and $M^{b}_{BB'}$ can be obtained analogously. Since any state $\rho_{AB}$ can be written in the same local basis as
\begin{equation}\Label{staterho}
\rho_{AB}=\sum_{i,j,k,l=1}^{d}a_{i,j,k,l}\ketbra{i}{j}\otimes\ketbra{k}{l}.
\end{equation}
Substituting Eqs.~\eqref{maxent} and \eqref{staterho} into Eq.~\eqref{piab}, we obtain the relation
\begin{equation}\Label{piab1}
\Pi_{ab}=\frac{1}{d_{A'}d_{B'}}U_{A'}^{a}\otimes U_{B'}^{b}\rho_{AB}^{T}U_{A'}^{a\dagger}\otimes U_{B'}^{b\dagger}
\end{equation}
which indicates that $\Pi_{ab}$ has the eigenvalues proportional to the eigenvalues of $\rho_{AB}$ (note that $\rho_{AB}^{T}$ here is a density matrix in $\mathcal{L}_{A'B'}$). 

Now we consider the case where the measurement operators take the form of more general pure states. Let the measurement operator $M_{A'A}^{a}=\ketbra{\psi_{a}}{\psi_{a}}$, where $\ket{\psi_{a}}\in\mathcal{H}_{A'}^{d}\otimes\mathcal{H}_{A}^{d}$ is an arbitrary state with Schmidt rank $d_{A}$. Since $\ket{\psi_{a}}$ can be written in Schmidt decomposition form as
\begin{equation}
    \ket{\psi_{a}}=\sum_{i=1}^{d_A}a_{i} \ket{u_{i}}_{A'}\ket{v_i}_{A},
\end{equation}
it is always possible to find a proper generalized Bell state $\ket{\Lambda_a}_{\psi}=\frac{1}{\sqrt{d_{A}}}\sum_{i}^{d_{A}}\ket{u_{i}}_{A'}\ket{v_i}_{A}$ and a positive semi-definite matrix $X_{a}\in\mathcal{L}_{A'}$ with the form
\begin{equation}
  X_{a}=x_{i}\ket{u_{i}}_{A'}\bra{u_{i}}_{A'},  
\end{equation}
where $x_{i}=\sqrt{d_{A}}a_{i}$, such that
\begin{equation}
  \ket{\psi_{a}}=X_{a}\otimes I_{A}\ket{\Lambda_a}_{\psi},  
\end{equation}
Since $\ket{\psi_a}_{\psi}$ has Schmidt rank $d_{A}$, that is, all $a_{i}$ are nonzero, $X_{a}$ is invertible. Analogously we can obtain invertible positive semi-definite matrix $X_{b}$ for a measurement operator $M_{BB'}^{b}$ which has Schmidt rank $d_{B}$. Then we have
\begin{equation}\Label{piab2}
\Pi_{ab}=X_{a}U_{A'}^{a}\otimes X_{b}U_{B'}^{b} \rho_{AB}^{T}U_{A'}^{a\dagger}X_{a}^{\dagger}\otimes U_{B'}^{b\dagger}X_{b}^{\dagger},
\end{equation}
where $U_{A'}^{a}$ is the local unitary that transforms a maximally entangled state $\ket{\Phi_{A'A}}=\frac{1}{\sqrt{d_A}}\sum_{i}^{d_A}\ket{v_i}\ket{v_i}$ into generalized Bell state $\ket{\Lambda_a}_{\psi}$, i.e. $\ket{\Lambda_a}_{\psi}=U_{A'}^{a}\otimes I_{A}\ket{\Phi_{A'A}}$. $U_{B'}^{b}$ is defined analogously.
In such a case, $\Pi_{ab}\geq 0$ if and only if $\rho_{AB}\geq 0$.

Since the reconstructed $X_{AB}$ will be exactly the same as $\Pi_{ab}$ when $\{\tau_{x,A'}\}_x$ and 
$\{\tau_{y,B'}\}_y$ are tomographically complete sets. When $\rho_{AB}$ is beyond-quantum, $\Pi_{ab}$ has a negative eigenvalue, which means that $X_{ab}^-$ must have positive eigenvalues, and we will obtain $c_{ab}> 0$.\qed

A natural question would be: is it always necessary to use tomographically complete input sets to attain universal completeness? We will discuss this question in the next section.

\section{Necessary condition for universal completeness}\label{S5}
In this section we show that it is necessary to use tomographically complete input sets on the composite system ${\cal H}_A\otimes {\cal H}_B$
to attain universal completeness 
when the local dimensions of Alice and Bob is the same $(d_{A}=d_{B})$. 
The local dimension is denoted by $d$. 
We will construct a special subset $S$ of beyond-quantum states,
for which the protocol is complete if and only if the input sets is tomographically complete. That is, it is impossible to detect all states in $S$ unless the input set is tomographically complete.

In our protocol, instead of directly dealing with the tested state, we witness the beyond-quantumness of the effective states $\Pi_{ab}$ with certain input set is $\{\tau_{x,A'}\otimes\tau_{y,B'}\}_{x,y}$ using SDP \eqref{SDP}. Therefore the minimum number of input states needed to detect a beyond-quantum state is determined by the minimum number of input states needed to detect the negativity of at least one effect state. More formally, 
we say that 
an input set $\{\tau_{x,A'}\otimes\tau_{y,B'}\}_{x,y}$ detects
a beyond-quantum state $\rho_{AB}$ 
when there exist two positive semi-definite operators 
$M^a_{A'A}\in\mathcal{L}_{A'A}$ and $M^b_{BB'}\in\mathcal{L}_{BB'}$ 
such that the effective state $\Pi_{ab}$ generated by Eq.~\eqref{piab} satisfies
the following.
(i) $\Pi_{ab}$ is a (unnormalized) beyond-quantum state. 
(ii) There exists a witness $Y$ of $\Pi_{ab}$, i.e., 
a Hermitian matrix $Y$ satisfying \eqref{witnessY},
with the form 
\begin{equation}\label{Ycondition}
 Y=\sum_{x,y} \alpha_{xy} \tau_{x,A'}\otimes \tau_{y,B'}
\end{equation}
with certain coefficients $\alpha_{xy}$. 
Conversely, a beyond-quantum state $\rho_{AB}$ can not be detected if 
no effective state $\Pi_{ab}$ generated by $\rho_{AB}$ 
satisfies the above two conditions.
For the simplicity of the following discussion, we first define two terminologies.

\begin{defi}
A family $S$ of beyond-quantum states is called demanding 
when the following condition holds: 
an input set $\{\tau_{x,A'}\otimes\tau_{y,B'}\}_{x,y}$ 
detects any state $\rho_{AB} \in S$
if and only if the input set $\{\tau_{x,A'}\otimes\tau_{y,B'}\}_{x,y}$ 
is tomographically complete. 
\end{defi}

\begin{defi}
A family $S$ of beyond-quantum states
is called witness-demanding if the following condition holds: 
For any element $\rho_{AB}\in S$, there exists
a witness $Y$ of $\rho_{AB}$ to satisfy Eq.~\eqref{Ycondition} 
if and only if the input set $\{\tau_{x,A}\otimes\tau_{y,B}\}_{x,y}$ is tomographically complete. 
\end{defi}

We discuss the relation between demanding and witness-demanding as following.

\begin{theorem}\Label{Th5}
If a family $S$ of beyond-quantum states 
is witness-demanding, it is also demanding.
\end{theorem}

\begin{proof}
Suppose $S$ is not demanding. 
Then, there exists a tomographically
\emph{incomplete} input set $\{\tau_{x,A'}\otimes\tau_{y,B'}\}_{x,y}$
to detect any element $\rho_{AB}\in S$.
That is, for every $\rho_{AB}\in S$, there exist two positive semi-definite operators 
$M_{A'A}^{a}\in\mathcal{L}_{A'A}$ and $M_{BB'}^{b}\in\mathcal{L}_{BB'}$ 
such that the effective state $\Pi_{ab}$ generated via Eq.~\eqref{piab} satisfies
the following:
(i) $\Pi_{ab}$ is a (unnormalized) beyond-quantum state. 
(ii) There exists a witness $Y$ of $\Pi_{ab}$ which satisfies \eqref{Ycondition} with the input set $\{\tau_{x,A'}\otimes\tau_{y,B'}\}_{x,y}$.

For operators $M_{A'A}^{a}\in\mathcal{L}_{A'A}$, we can define the CP map $L_{AA'}^{a}$ from $\mathcal{L}_h(A)$ to $\mathcal{L}_h(A')$ satisfying
\begin{equation}
 L_{AA'}^{a}(W_{A})=\Tr_{A} M_{A'A}^{a} (I_{A'}\otimes W_{A}) 
\end{equation}
for $W_{A}\in \mathcal{L}_{h}(A)$.
$L_{BB'}^{b}$ can be defined analogously. 
We also define 
$L_{AA'}^{a*}$ and $L_{BB'}^{b*}$ as the adjoint maps of 
$L_{AA'}^{a}$ and $L_{BB'}^{b}$ respectively.

Then for every $\rho_{AB}\in S$, we can define the Hermitian matrix 
$Y^{*}:=\sum_{x,y} \alpha_{xy}L_{AA'}^{a*}(\tau_{x,A'})\otimes L_{BB'}^{b*}(\tau_{y,B'})$. And we have 
\begin{equation}\label{witcomp1}
\begin{split}
&\Tr \rho_{AB}Y^{*}\\
=&\sum_{x,y} \alpha_{xy} \Tr [\rho_{AB} L_{AA'}^{a*}(\tau_{x,A'})\otimes L_{BB'}^{b*}(\tau_{y,B'})] \\
=&\sum_{x,y} \alpha_{xy} \Tr [L_{AA'}^{a}\otimes 
 L_{BB'}^{b}(\rho_{AB}) \tau_{x,A'}\otimes \tau_{y,B'}]\\
=& \Tr \Pi_{ab}Y <0.
 \end{split}
\end{equation}
Eq.~\eqref{witcomp1} shows that $Y^{*}$ is a witness of $\rho_{AB}$. However, the set $\{L_{AA'}^{*}(\tau_{x,A'})\otimes L_{BB'}^{*}(\tau_{y,B'})\}_{x,y}$ is tomographically \emph{incomplete} on 
$\mathcal{L}_{A'B'}$ since $\{\tau_{x,A'}\otimes\tau_{y,B'}\}_{x,y}$ is tomographically \emph{incomplete}.
Hence, the relation \eqref{witcomp1} contradicts the assumption that $S$ is witness-demanding. Therefore $\rho_{AB}$ must be demanding.
\end{proof}


We define the set ${\cal P}(\cH_{A'}\otimes \cH_{B'})$
of pure states $|\psi\rangle$ on $\cH_{A'}\otimes \cH_{B'}$
whose maximum Schmidt coefficient $\sc(\rho)$
is not greater than $\sqrt{3/2 d}$.
\begin{lemma}\label{BML2}
The linear space spanned by
${\cal P}(\cH_{A'}\otimes \cH_{B'})$
equals the set ${\cal L}_h(A'B')$ of Hermitian matrices on $\cH_{A'}\otimes \cH_{B'}$.
\end{lemma}

\begin{proof}
We define the operators $X,Z$ on $\cH_{A'}$ as
\begin{align}
X:= \sum_{j=0}^{d-1} |j+1\rangle \langle j|,\quad
Z:= \sum_{j=0}^{d-1} \omega^j |j\rangle \langle j|,
\end{align}
with $\omega :=e^{2\pi i/d}$.
We define $ |\Phi\rangle:=
\sum_{j=0}^{d-1}\frac{1}{\sqrt{d}}|jj\rangle$
and
$|\Phi_{j,k}\rangle:=
X^j Z^k
\otimes I_{B'} |\Phi\rangle$.
For $(j,k) \neq (j',k')$ and $c \in (0,1)$, we define 
\begin{align}
|\Phi_{j,k,j',k'|\sqrt{c}}\rangle
:=&
\sqrt{1-c} |\Phi_{j,k}\rangle
+\sqrt{c} |\Phi_{j',k'}\rangle \\
|\Phi_{j,k,j',k'|i \sqrt{c}}\rangle
:=&
\sqrt{1-c} |\Phi_{j,k}\rangle
+i \sqrt{c} |\Phi_{j',k'}\rangle.
\end{align}

The set $\{|\Phi_{j,k}\rangle \}$ forms an orthogonal basis of $\cH_{A'}\otimes \cH_{B'}$.
Hence, it is sufficient to show the following:
for $j,k,j',k'$,
there exist real numbers $c \in (0,1)$ such that
\begin{align}
&|\Phi_{j,k,j',k'|\sqrt{c}}\rangle
\langle \Phi_{j,k,j',k'|\sqrt{c}}|,
|\Phi_{j,k,j',k'|i\sqrt{c}}\rangle
\langle \Phi_{j,k,j',k'|i \sqrt{c}}|\nonumber \\
&\in {\cal P}(\cH_{A'}\otimes \cH_{B'}). \label{NM1}
\end{align}
More specifically, 
since $\{|\Phi_{j,k}\rangle \}_{j,k}$ 
forms an orthogonal basis of $\cH_{A'}\otimes \cH_{B'}$, and $\{\ketbra{\Phi_{j,k}}{\Phi_{j,k}}\}_{j,k}\subset {\cal P}(\cH_{A'}\otimes \cH_{B'})$,
the set ${\cal D}(A'B')$
of the diagonal matrices in $\mathcal{L}_{h}(A'B')$ are included in the linear span of 
${\cal P}(\cH_{A'}\otimes \cH_{B'})$. The quotient space ${\cal L}_h(A'B')/{\cal D}(A'B')$
 is spanned by 
$\{ |\Phi_{j',k'}\rangle\langle \Phi_{j,k}|+|\Phi_{j,k}\rangle\langle \Phi_{j',k'}|\}_{
j,k,j',k'}
\cup
\{ i|\Phi_{j',k'}\rangle\langle \Phi_{j,k}|-i|\Phi_{j,k}\rangle\langle \Phi_{j',k'}|\}_{
j,k,j',k'}$. Since the quotient space ${\cal L}_h(A'B')/{\cal D}(A'B')$ includes all off-diagonal elements, the other matrices in $\mathcal{L}_{h}(A'B')$ are equivalent to the matrices in \eqref{NM1}
in the sense of this quotient space. Therefore $\mathcal{L}_{h}(A'B')$ equals to the linear span of 
${\cal P}(\cH_{A'}\otimes \cH_{B'})$ if \eqref{NM1} is satisfied.

We choose $c \in (0,1)$ as
$\sqrt{c}+\sqrt{1-c}=\sqrt{\frac{3}{2}}$.
Then, we show that
the maximum Schmidt coefficients of
$|\Phi_{j,k,j',k'|\sqrt{c}}\rangle$
and $|\Phi_{j,k,j',k'|i \sqrt{c}}\rangle$
are upper bounded by $(\sqrt{c}+\sqrt{1-c})/\sqrt{d}= \sqrt{3/2 d}$
as follows.
We choose 
the simultaneous Schmidt basis
$\{|v_{l,A}\rangle\}_{l=0}^{d-1}$, 
$\{|v_{l,B}\rangle\}_{l=0}^{d-1}$ for
the vectors
$|\Phi_{j,k}\rangle$ and $|\Phi_{j',k'}\rangle$ as \cite{hiroshima2004finding}
\begin{align}
|\Phi_{j,k}\rangle
=&\frac{1}{\sqrt{d}}
\sum_{l=0}^{d-1}e^{i\theta_1(l)} |v_{l,A},v_{l,B}\rangle ,\\
|\Phi_{j',k'}\rangle=&\frac{1}{\sqrt{d}}
\sum_{l=0}^{d-1}e^{i\theta_2(l)} |v_{l,A},v_{l,B}\rangle .
\end{align}
Although $\theta_1(l)$ and $\theta_2(l)$ depend on 
$j,k,j',k'$, we omit $j,k,j',k'$ in the expressions for simplicity.
Hence,
\begin{align}
&|\Phi_{j,k,j',k'|\sqrt{c}}\rangle\nonumber \\
=&\frac{1}{\sqrt{d}}\sum_{l=0}^{d-1}
(\sqrt{1-c}e^{i\theta_1(l)}+ \sqrt{c}e^{i\theta_2(l)})
 |v_{l,A},v_{l,B}\rangle \\
&|\Phi_{j,k,j',k'|i\sqrt{c}}\rangle\nonumber \\
=&\frac{1}{\sqrt{d}}\sum_{l=0}^{d-1}
(\sqrt{1-c}e^{i\theta_1(l)}+i \sqrt{c}e^{i\theta_2(l)})
 |v_{l,A},v_{l,B}\rangle .
\end{align}
Since $
|\sqrt{1-c}e^{i\theta_1(l)}+ \sqrt{c}e^{i\theta_2(l)}|
\le \sqrt{c}+\sqrt{1-c} =\sqrt{\frac{3}{2}}$,
the maximum Schmidt coefficient of
$|\Phi_{j,k,j',k'|\sqrt{c}}\rangle$
is upper bounded by $ \sqrt{3/2 d}$.
Similarly, 
since $|\sqrt{1-c}e^{i\theta_1(l)}+ i\sqrt{c}e^{i\theta_2(l)}|
\le \sqrt{\frac{3}{2}}$,
the maximum Schmidt coefficient of
$|\Phi_{j,k,j',k'|i\sqrt{c}}\rangle$
is upper bounded by $ \sqrt{3/2 d}$.
Hence, we have shown that condition \eqref{NM1} can be satisfied. This completes the proof of Lemma \ref{BML2}.
\end{proof}

For $ \rho \in {\cal P}(\cH_{A'}\otimes \cH_{B'})$ and $0<t< 2d/3-1$,
we define $\Pi^t(\rho):=
\frac{1}{d^2}I- \frac{1+t}{d^2}\rho$.
Then, $\Pi^t(\rho)$ is a beyond-quantum state because
any two positive semi-definite matrices $M_{A'}$ and $ M_{B'}$
satisfy
\begin{align}
&\Tr (M_{A'}\otimes M_{B'}) \Pi^t(\rho)\notag\\
\ge &
\frac{1}{d^2}(\Tr M_{A'}) (\Tr M_{B'})
- \frac{1+t}{d^2}\sc(d)^2
(\Tr M_{A'})(\Tr M_{B'})\notag\\
=&
\frac{1}{d^2}(1 - (1+t) \sc(d)^2)
(\Tr M_{A'})
(\Tr M_{B'})\notag\\
\ge &
\frac{1}{d^2}( 1-  (1+t) \cdot \frac{3}{2d})
(\Tr M_{A'})
(\Tr M_{B'})
> 0.
\end{align}

\begin{lemma}\label{BML1}
We choose $\rho \in {\cal P}(\cH_{A'}\otimes \cH_{B'})$.
When an input set $\{\tau_{x,A'}\otimes\tau_{y,B'}\}_{x,y}$ 
detects any state 
$\rho_{A'B'} \in \{\Pi^t(\rho)\}_{0<t< 2d/3-1}$,
there exist coefficients $c_\rho(x,y)$ such that
$\rho= \sum_{x,y}c_\rho(x,y)\tau_{x,A'}\otimes\tau_{y,B'}$.
\end{lemma}

\begin{proof}
For any positive semi-definite matrix $Y$,
the relation $\Tr \Pi^0(\rho) Y=0$ holds only when 
$Y$ is a constant times of $\rho$.
When $t$ is an infinitesimal small real number $t>0$
a witness of $\Pi^t(\rho)$ is limited to 
a constant times of $\rho$.
Hence, the desired statement is obtained.
\end{proof}

\begin{theorem}
The family
$\{\Pi^t(\rho)\}_{\rho \in {\cal P}(\cH_{A'}\otimes \cH_{B'}),
0<t< 2d/3-1}$ of 
beyond-quantum states is witness-demanding.
\end{theorem}

Therefore, the family
$\{\Pi^t(\rho)\}_{\rho \in {\cal P}(\cH_{A'}\otimes \cH_{B'}),
0<t< 2d/3-1}$ of beyond-quantum states is demanding.
That is,
to realize universal completeness,
the input set $\{\tau_{x,A'}\otimes\tau_{y,B'}\}_{x,y}$ 
needs to be tomographically complete.

\begin{proof}
We choose an input set $\{\tau_{x,A'}\otimes\tau_{y,B'}\}_{x,y}$ that
detects any state $\rho_{A'B'} $ in
$\{\Pi^t(\rho)\}_{\rho \in {\cal P}(\cH_{A'}\otimes \cH_{B'}),
0<t< 2d/3-1}$.
Due to Lemma \ref{BML1},
for $\rho \in {\cal P}(\cH_{A'}\otimes \cH_{B'})$,
there exist coefficients $c_\rho(x,y)$ such that
$\rho= \sum_{x,y}c_\rho(x,y)\tau_{x,A'}\otimes\tau_{y,B'}$.
Due to Lemma \ref{BML2},
the input set $\{\tau_{x,A'}\otimes\tau_{y,B'}\}_{x,y}$ 
spans the of Hermitian matrices on $\cH_{A'}\otimes \cH_{B'}$.
\end{proof}

Now we discuss how many input states are required to at least detect one beyond-quantum state, and how to choose the input states in such a circumstance. We present a result in 2-dimensional case.

\begin{theorem}\Label{Th6} 
Let $d_{A}=d_{B}=2$, and let $\{\cal{Z}\}$ be the set of all matrices which are unitarily equivalent to the Pauli matrix $\sigma_{z}$.  
If there exists a matrix $Z\in\cal{Z}$ such that $\Tr(Z\tau_{x,A'})=0$ for all input states $\tau_{x,A'}$, no beyond-quantum state $\rho_{AB}\in \mathcal{S}_{AB}$ can be detected by the MDI protocol. The theorem works analogously for input states $\tau_{y,B'}$.
\end{theorem}

Theorem \ref{Th6} implies that, in order that our MDI protocol detects at least one beyond-quantum state with $d_A=d_B=2$,
Alice and Bob need to input at least 3 quantum states, respectively.
Additionally, the chosen input states should be linearly independent and non-orthogonal.

\section{conclusion and discussion}\label{S6}
In this paper, we have proposed a practical MDI test to detect beyond-quantum states, which further shows the power of MDI protocol on such a task. 
Compared to the previous work \cite{lobo2022certifying} which also applies a MDI Bell test for the same task, 
our method processes the experimental correlations in an optimal way, therefore greatly improves the detectability in practical scenarios where the form of the tested state, the measurement operators,
nor the input states is not known. 
More specifically, our method attains the universal completeness only with the natural assumption of knowing the local dimensions of the tested system, instead of the tested states. 

Also, due to the advantage that the detection result of our protocol is independent of the form of tested state and measurement operators, 
we are able to study the importance of having tomographically complete quantum input set. 
We have shown that, under the typical case where $d_{A}=d_{B}$,
the universal completeness can be attained 
if and only if the input sets are tomographically complete. 
We have also clarified the importance of  
the choice of the input states 
when the input set is tomographically incomplete. 
We have presented an example in $d_{A}=d_{B}=2$ cases to support this importance. All the above results are possible to be extended into continuous variable case based on the method proposed in Ref.~\cite{abiuso2021measurement}. We leave such extension for further work.

Further, as a byproduct,
using the relation between entangled states and beyond-quantum states, 
we can derive an interesting observation for 
MDI entanglement detection as follows.
In $d_{A}=d_{B}=2$ case, 
any beyond-quantum state $\rho_{AB}^{bq}$ 
can be described by using 
an entangled state $\rho_{AB}^{E}$ 
with partial transposition, namely, $\rho_{AB}^{bq}=\Gamma_{A}\otimes I_{B}(\rho_{AB}^{E})$. 
Then, their respective effective states $\Pi^{bq}_{ab}$ and $\Pi_{ab}^{E}$ have the relation 
$\Pi^{bq}_{ab}=\Gamma_{A}\otimes I_{B}(\Pi_{ab}^{E})$
under the same measurement operators $M_{A'A}^{a}$ and $M_{BB'}^{b}$.
Therefore, 
universal completeness in entangled state detection implies universal completeness in beyond-quantum state detection.
This implication yields that
the input set needs to be tomographically complete to attain universal completeness in entangled state detection
in the $d_{A}=d_{B}=2$ case.
A similar kind of observation can be expected in more general cases.

\section*{Acknowledgement}
We thank Hayato Arai and Paolo Abiuso for their valuable comments. MH was supported in part by the National
Natural Science Foundation of China under Grant 62171212.

\bibliographystyle{apsrev4-1} 
\bibliography{ref}

\appendix

\section{Proof that LOCC will not produce false detection}\label{appendixA}
Since a multi-round LOCC is composed of one-way LOCC \cite{chitambar2014everything}, it is sufficient to prove that any one-way LOCC does not produce false detection when the tested state is quantum. Let $\mathcal{B}(\mathcal{H}_{X})$ denote the set of bounded linear operators on $\mathcal{H}_{X}$. Any one-way LOCC operated on a tested system $\mathcal{B}(\mathcal{H}_{AB})$ can be represented as a quantum instrument composed of $n$ completely positive (CP) maps $\{\mathcal{F}_{i}\}_{i=1,2,...,n}$ satisfying the following conditions:

(i) $\mathcal{F}_{i}=\mathcal{E}_{i}^{A}\otimes\mathcal{P}_{i}^{B}$, where $\mathcal{E}_{i}^{A}$ is a CP map on $\mathcal{B}(\mathcal{H}_{A})$ and $\mathcal{P}_{i}^{B}$ is a completely positive trace-preserving (CPTP) map on $\mathcal{B}(\mathcal{H}_{B})$. Note that the condition applies to the case where the communication is from Alice to Bob. In the opposite direction case, $\mathcal{F}_{i}=\mathcal{P}_{i}^{A}\otimes\mathcal{E}_{i}^{B}$.

(ii) $\sum_{i=1}^{n}\mathcal{F}_{i}$ is a CPTP map.\\
For such a one-way LOCC, we have the following relation: 
\begin{equation}
\begin{split}
  &\Tr [\mathcal{E}^{A'A}_{i}\otimes\mathcal{F}^{BB'}_{j}(I\otimes\rho_{AB}\otimes I)\cdot  M_{A'A}^{a}\otimes M_{BB'}^{b}]\\
  =&\Tr [\tau_{x}\otimes\rho_{AB}\otimes\tau_{y}\cdot  \mathcal{E}^{A'A*}_{i}\otimes\mathcal{F}^{BB'*}_{j}(M_{A'A}^{a}\otimes M_{BB'}^{b})]\\
  =&\Tr (\tau_{x}\otimes\rho_{AB}\otimes\tau_{y}\cdot  \hat{M}_{A'A}^{a(i)}\otimes \hat{M}_{BB'}^{b(j)})
\end{split}
\end{equation}
where $\mathcal{E}^{A'A*}_{i}$ and $\mathcal{F}^{BB'*}_{j}$ are the adjoint CP maps of $\mathcal{E}^{A'A}_{i}$ and $\mathcal{F}^{BB'}_{j}$, they map POVM operators $M_{A'A}^{a}$ 
and $M_{BB'}^{b}$ into $\hat{M}_{A'A}^{a(i)}$ and $\hat{M}_{BB'}^{b(j)}$, which are also POVM operators. Therefore the correlations generated with such a one-way LOCC can be represented as 
\begin{equation}
    p(a,b|x,y)=\Tr(\hat{\Pi}_{ab}^{i,j}\tau_{x}\otimes\tau_{y}),
\end{equation}
where 
\begin{equation}\label{positivePi}
  \hat{\Pi}_{ab}^{i,j}=\Tr_{AB} (I\otimes \rho_{AB}\otimes I \cdot\hat{M}_{A'A}^{a(i)}\otimes \hat{M}_{BB'}^{b(j)}).
\end{equation}
Eq.~\eqref{positivePi} shows that when $\rho_{AB}$ is quantum,  $\hat{\Pi}_{ab}^{i,j}$ is positive semi-definite for every pair of outcomes $(a,b)$. In such case, the solution $\overline{c_{ab}}$ of SDP \eqref{SDP} will always be zero. Therefore LOCC will not produce false detection when $\rho_{AB}$ is quantum. 
 
\section{Proof of theorem \ref{Th6}}   
Let $\Gamma_{A'}$ be the partial transposition on $A'$.
Then, we prepare the following lemma.
\begin{lemma}\label{LL3}
When the dimensions of systems $A$ and $B$ are $d_{A}=d_{B}=2$, 
any beyond-quantum state $\rho_{AB}^{bq}$ in $\mathcal{S}_{AB}$
can be written with a form
$\rho_{AB}^{bq}= \Gamma_{A}(\rho_1)+\rho_2$, where
$\rho_1$ and $\rho_2$ are positive semi-definite matrices on $\cH_{A'B'}$.
\end{lemma}

The proof of lemma \ref{LL3} is essentially shown in Proposition 11 of Ref.~\cite{arai2023detection}, where it is proved in $2\times 2$-dimensional case, the set of all extremal beyond-quantum states is given as $\{\Gamma(\rho)|\rho \text{ is pure entangled quantum states}\}\cup\{\rho|\rho \text{ is pure quantum states}\}$. 

Now we come to the proof of theorem \ref{Th6}.

\emph{Proof of theorem \ref{Th6}.} 
As the assumption, we choose 
a matrix $Z\in\cal{Z}$ 
such that $\Tr(Z\tau_{x,A'})=0$ for all input states $\tau_{x,A'}$.
Matrix $Z$ can be represented as $Z=U \sigma_2 U^\dagger $ with a properly chosen local unitary $U$ on $\cH_{A'}$
 (note that we denote Pauli matrices $I$,$\sigma_{x}$,$\sigma_{y}$,$\sigma_{z}$ by $\sigma_{i}\ (i=0,1,2,3)$ for simplicity).
Given a beyond-quantum state $\rho_{AB}^{bq}\in \mathcal{S}_{AB}$
and measurement operators $M_{A'A}^{a}$ and $M_{BB'}^{b}$,
a corresponding effective state 
$\Pi_{ab}\in \cal{L}_{A'B'}$ is defined via Eq.~\eqref{piab}.
In our MDI protocol, we use SDP \eqref{SDP} to find the least negative Hermitian matrix $X_{ab}$ satisfying
\begin{equation}\label{Pcond}
\Tr(\tau_{x,A'}\otimes \tau_{y,B'} X_{ab})=\Tr(\tau_{x,A'}\otimes \tau_{y,B'} \Pi_{ab})
\end{equation}
for all $\tau_{x,A'}$ and $\tau_{y,B'}$.
If for any possible effective state $\Pi_{ab}$ generated by $\rho_{AB}^{bq}$, there exists a positive semi-definite matrix $X_{ab}$ satisfying condition \eqref{Pcond}, the beyond-quantumness of $\rho_{AB}^{bq}$ can not be detected. 

When $\Pi_{ab}$ is beyond-quantum, 
Lemma \ref{LL3} guarantees that
there exist two positive semi-definite matrices $W_{ab,1}$ and 
$W_{ab,2}$ on $\cH_{A'B'}$ such that
$(U^\dagger\otimes I_B)\Pi_{ab} (U\otimes I_B)= \Gamma_{A}(W_{ab,1})+W_{ab,2}$.
The positive semi-definite matrix $W_{ab,1}$ can be decomposed 
as
\begin{equation}\Label{rhoq}
W_{ab,1}
=\sum_{kj=0}^{3} g_{kj}\sigma_{k}\otimes\sigma_{j}.
\end{equation}
with real coefficients $g_{ij}$. 
Hence, 
\begin{equation}\Label{rhoq}
\Gamma_{A}(W_{ab,1})
=\sum_{k=0,1,3}\sum_{j=0}^{3} g_{kj}\sigma_{k}\otimes\sigma_{j}
-2\sum_{j=0}^{3} g_{2j}\sigma_{2}\otimes\sigma_{j}.
\end{equation}
Then, 
the Hermitian matrix $\hat{\Pi}_{ab}:=\Pi_{ab}
+2(U \otimes I_B)\sum_{j=0}^{3} g_{2j}\sigma_{2}\otimes\sigma_{j} (U^\dagger\otimes I_B)$ is shown to be positive semi-definite as follows.
\begin{align*}
&(U^\dagger \otimes I_B)\hat{\Pi}_{ab}(U \otimes I_B) \\
= &\Gamma_{A}(W_{ab,1})
+2 \sum_{j=0}^{3} g_{2j}\sigma_{2}\otimes\sigma_{j}+W_{AB,2}\\
=& W_{AB,1}+W_{AB,2} \ge 0.
\end{align*}
Since 
\begin{align*}
&\Tr (\tau_{x,A'}\otimes \tau_{y,B'})
2(U \otimes I_B)\sum_{j=0}^{3} g_{2j}\sigma_{2}\otimes\sigma_{j} (U^\dagger\otimes I_B)\\
=&
2 \sum_{j=0}^{3} g_{2j}
(\Tr   \tau_{x,A'} U \sigma_{2}U^\dagger)
(\Tr \tau_{y,B'} \sigma_{j}) \\
=&
2 \sum_{j=0}^{3} g_{2j}
(\Tr   \tau_{x,A'} Z)
(\Tr \tau_{y,B'} \sigma_{j}) 
=0,
\end{align*}
we have
\begin{align}
\Tr(\tau_{x,A'}\otimes \tau_{y,B'} \hat{\Pi}_{ab})
=
\Tr(\tau_{x,A'}\otimes \tau_{y,B'} \Pi_{ab})
\end{align}
for all $\tau_{x,A'}$ and $\tau_{y,B'}$.
Therefore, the beyond-quantum state $\rho_{AB}^{bq}\in \mathcal{S}_{AB}$ can not be detected by the MDI protocol in this case.
\qed


\end{document}